%% file: strong.tex
\pdfoutput=1
\documentclass[online]{tlp}
\usepackage{amsfonts}
\usepackage{latexsym}
\usepackage{amssymb}
\usepackage{amsmath}
\usepackage{microtype}
\usepackage[arxiv]{preamble}

\newcommand{\coNP}{$\mathbf{coNP}$}

\let\cite\citep

\begin{document}

\lefttitle{A. Charalambidis, Ch. Nomikos and P. Rondogiannis}

\jnlPage{1}{8}
\jnlDoiYr{2021}
\doival{10.1017/xxxxx}

\title[Strong Equivalence of LPODs: a Logical Perspective]
  {Strong Equivalence of Logic Programs with Ordered Disjunction: a Logical Perspective}


\begin{authgrp}
\author{\gn{Angelos} \sn{Charalambidis}}
\affiliation{Dept of Informatics and Telecommunications, National and Kapodistrian University of Athens, Greece}
\author{\gn{Christos} \sn{Nomikos}}
\affiliation{Dept of Computer Science and Engineering, University of Ioannina, Greece}
\author{\gn{Panos} \sn{Rondogiannis}}
\affiliation{Dept of Informatics and Telecommunications, National and Kapodistrian University of Athens, Greece}
\end{authgrp}

\history{\sub{xx xx xxxx;} \rev{xx xx xxxx;} \acc{xx xx xxxx}}

\maketitle

\begin{abstract}
Logic Programs with Ordered Disjunction (LPODs) extend classical logic programs
with the capability of expressing preferential disjunctions in the heads of
program rules. The initial semantics of LPODs~\cite{lpod-brewka,lpod-BNS04},
although simple and quite intuitive, is not purely model-theoretic. A
consequence of this is that certain properties of programs appear non-trivial to
formalize in purely logical terms. An example of this state of affairs is the
characterization of the notion of strong equivalence for LPODs~\cite{FaberTW08}.
Although the results of \citeN{FaberTW08} are accurately developed, they fall
short of characterizing strong equivalence of LPODs as logical equivalence in
some specific logic. This comes in sharp contrast with the well-known
characterization of strong equivalence for classical logic programs, which, as
proved by \citeN{LifschitzPV01}, coincides with logical equivalence in the logic
of here-and-there. In this paper we obtain a purely logical characterization of
strong equivalence of LPODs as logical equivalence in a four-valued logic.
Moreover, we provide a new proof of the \coNP-completeness of strong equivalence
for LPODs, which has an interest in its own right since it relies on the special
structure of such programs. Our results are based on the recent logical
semantics of LPODs introduced by \citeN{iclp2021}, a fact which we believe
indicates that this new semantics may prove to be a useful tool
in the further study of LPODs.
\ifarxiv
This work is under consideration for acceptance in TPLP.%
\fi
\end{abstract}
\begin{keywords}
Ordered Disjunction, Strong Equivalence, Logic of Here-and-There, Answer Sets.
\end{keywords}

\section{Introduction}
Logic Programs with Ordered Disjunction (LPODs)~\cite{lpod-brewka,lpod-BNS04}
extend classical logic programs with the capability of expressing preferential
disjunctions in the heads of program rules. The head of an LPOD rule is a formula
$C_1 \times \cdots \times C_n$ intuitively understood as follows: \qemph{I prefer $C_1$;
however, if $C_1$ is impossible, I can accept $C_2$; $\cdots$; if all of
$C_1,\ldots,C_{n-1}$ are impossible, I can accept $C_n$}. The meaning of LPODs
is expressed by their \emph{most-preferred answer sets}~\cite{lpod-brewka,lpod-BNS04},
namely a subset of their answer sets which satisfies in the best possible way the preferences
in the head of program rules. Due to their elegance and expressiveness, LPODs are
widely accepted as a concise and powerful formalism for preferential reasoning, both 
in logic programming and in artificial intelligence.

Although simple and quite intuitive, the original semantics of
LPODs~\cite{lpod-brewka,lpod-BNS04} is not purely model-theoretic. More
specifically, the most-preferred answer sets of a program can not be determined by
just examining the set of models of the program. Instead, one has to
additionally use an ordering relation which relies on the syntax of the source
program. There have been reported in the literature~\cite{lpod-crprolog,lpod-BNS04,iclp2021}
cases where the original semantics of LPODs produces counterintuitive results.
Another consequence of this semantics, is that certain
properties of LPODs appear non-trivial to formalize in purely logical terms.
In this paper we identify one such case, namely the problem of characterizing
the notion of \emph{strong equivalence} for LPODs.

The concept of strong equivalence for logic programs was introduced
by \citeN{LifschitzPV01}, and has proven to be an essential and extensively
studied property in ASP. Two logic programs $P_1$ and $P_2$ are termed strongly
equivalent under a given semantics if for every logic program $P$, $P_1 \cup P$
has the same meaning as $P_2 \cup P$ under this given semantics. Obviously, when
two logic programs are strongly equivalent, we can replace one for the other
inside a bigger program without any change in the observable behavior of this
program. \citeN{LifschitzPV01} demonstrated that two programs are
strongly equivalent under the answer set semantics~\cite{GL88} if and only if
they are equivalent in the logic of here-and-there~\cite{Pearce96,Pearce99}.
The importance of this result stems from the fact that it relates the observable
behaviour of programs with a purely logical notion, namely that of logical equivalence.

Due to the significance of strong equivalence, it appears as a natural
endeavor to study this concept for various extensions of logic programs.
Shortly after the inception of LPODs, an exhaustive study of various notions
of strong equivalence for LPODs was undertaken by \citeN{FaberTW08}. Although
the results of \citeN{FaberTW08} are accurately developed, they fall short of
characterizing strong equivalence of LPODs as logical equivalence in some
specific logic. This comes in sharp contrast with the aforementioned
characterization of strong equivalence for classical logic programs as logical
equivalence in the logic of here-and-there. We believe that this is not an
inherent shortcoming of the work of \citeN{FaberTW08}, but instead a possibly
unavoidable consequence of the fact that the original semantics of
LPODs is not purely model theoretic.

Recently, a purely model-theoretic semantics for LPODs was developed
by \citeN{iclp2021}, who undertook a question initially posed
by \citeN{lpod-cabalar}. More specifically, as it is demonstrated
by \citeN{iclp2021}, the most-preferred answer sets of an LPOD can be obtained as the
least models of the program under a novel four-valued logic, using an ordering
relation that is independent of the syntax of the program. It is also
demonstrated that the shortcomings of LPODs that have been observed in the
literature~\cite{lpod-crprolog,lpod-BNS04,iclp2021}, are remedied by resorting
to this new approach, and it is claimed that this new
semantics may prove helpful in formalizing, in purely logical terms, properties
and transformations of LPODs. It is therefore natural to wonder if this new
semantic characterization leads to a purely logical definition of strong
equivalence for LPODs. The present paper investigates exactly this question. More
specifically, the main contributions of the present paper are as follows:
\begin{itemize}
\item Following the work of~\citeN{FaberTW08}, we consider two alternative definitions
      of strong equivalence for LPODs, which can be supported under the
      model-theoretic framework developed by~\citeN{iclp2021}. We demonstrate
      that both of them coincide with the notion of logical equivalence of
      programs in the four-valued logic of~\citeN{iclp2021}. Our characterization
      gracefully extends the results of~\citeN{LifschitzPV01} for normal logic
      programs.

\item We provide a new proof of the \coNP-completeness of strong equivalence for
      LPODs, which has an interest in its own right, since it relies on the
      special structure of such programs. More specifically, the proof demonstrates
      \coNP-hardness by a direct (and quite simple) reduction from 3SAT, without
      resorting to the well-known (and more involved) \coNP-hardness result of~\citeN{Lin02}
      for strong equivalence of normal logic programs.
\end{itemize}
The rest of the paper is organized as follows. Section~\ref{background} provides
the mathematical preliminaries that will be needed throughout the paper.
Section~\ref{characterization} presents the characterization results for strong
equivalence of LPODs. In Section~\ref{complexity} the \coNP-completeness of
strong equivalence for LPODs is established. Section~\ref{related} discusses
related work and gives pointers for future work. The proofs of certain
results have been moved to an appendix.

\section{Background}\label{background}
In this section we present the necessary background that will be used throughout
the paper.
We start by defining the syntax and the semantics of the four-valued logic introduced
by \citeN{iclp2021}, and discuss how this logic can be used to
redefine the semantics of LPODs.

Similarly to the paper by \citeN{FaberTW08}, we do not consider strong negation, for reasons
of simplicity.
\begin{definition}\label{sigma-def}
Let $\Sigma$ be a nonempty, countably infinite, set of propositional atoms. The set
of \emph{well-formed formulas} is inductively defined as follows:
\begin{itemize}
\item Every element of $\Sigma$ is a well-formed formula,

\item If $\phi_1$ and $\phi_2$ are well-formed formulas, then
      $(\phi_1 \wedge \phi_2)$, $(\phi_1 \vee \phi_2)$,
      $(\pnot \phi_1)$, $(\phi_1 \leftarrow \phi_2)$,
      and $(\phi_1 \times \phi_2)$, are well-formed formulas.
\end{itemize}
\end{definition}

We will use capital variables, like $A$, $B$, $C$, $D$, and their subscripted versions,
to denote atoms; we will use $L$, and its subscripted versions, to denote literals
(namely, atoms or negated atoms).

In order to define the semantics of well-formed formulas, we use the set
$V = \{F, F^*, T^*, T\}$ of truth values, which are ordered as follows:
\[
          F < F^* < T^* < T
\]
\pagebreak[4]
\begin{definition}\label{interpretation-and-semantics}
An \emph{interpretation} $I$ is a function from $\Sigma$ to $V$.
We can extend $I$ to apply to formulas, as follows:
\[
\begin{array}{lll}
    I(\pnot\phi)  & = &
      \begin{cases}
          T & \mbox{if $I(\phi)\leq F^*$}\\
          F & \mbox{otherwise}
      \end{cases} \\
    I(\phi \leftarrow \psi)  & = &
      \begin{cases}
          T & \mbox{if $I(\phi) \geq I(\psi)$}\\
          F & \mbox{otherwise}
      \end{cases} \\
I(\phi_1 \wedge \phi_2) & = & \min\{I(\phi_1),I(\phi_2)\}\\
I(\phi_1 \vee   \phi_2) & = & \max\{I(\phi_1),I(\phi_2)\}\\
I(\phi_1 \times \phi_2) & = &
      \begin{cases}
          I(\phi_2) & \mbox{if $I(\phi_1) = F^*$}\\
          I(\phi_1) & \mbox{otherwise}
      \end{cases} \\
\end{array}
\]
\end{definition}
It is straightforward to see that the meanings of ``$\vee$'', ``$\wedge$'', and
``$\times$'' are associative and therefore we can write $I(\phi_1 \vee \cdots
\vee \phi_n)$, $I(\phi_1 \wedge \cdots \wedge \phi_n)$, and $I(\phi_1 \times
\cdots \times \phi_n)$ unambiguously (without the need of extra parentheses).
Moreover, given literals $L_1,\ldots,L_n$, we will often write
$L_1,\ldots,L_n$ instead of $L_1\wedge \cdots \wedge L_n$.

LPODs are sets of formulas of a special kind, specified by the following definition.
\begin{definition}\label{lpod}
An LPOD is a finite set of rules of the form:
\[
  C_1 \times \cdots \times C_n \leftarrow A_1,\ldots,A_m,{\pnot B_1},\ldots,{\pnot B_k}
\]
where $n\geq 1$, $m,k\geq 0$, and the $C_i, A_j$, and $B_l$ are atoms.
\end{definition}
We will use capital letters like $P$, $Q$, and their subscripted versions, to denote LPODs.
\begin{definition}\label{model}
An interpretation $I$ is a \emph{model} of an LPOD $P$ if every rule of $P$
evaluates to $T$ under $I$. Two LPODs are termed \emph{logically equivalent}
if they have the same models.
\end{definition}

\citeN{iclp2021} defined the semantics of LPODs, namely the precise characterization
of their most-preferred answer sets, based on the above four-valued
logic. More specifically, the most-preferred answer sets of an LPOD are generated
using a two-step procedure. In the first step, a subset of the models of the program
is selected using a minimization procedure according to an ordering relation
$\preceq$ defined below. These models are called \emph{answer sets}
of the given LPOD, because they can also be produced using a reduct-based approach
similar to the one defined in the paper by \citeN{lpod-BNS04}. In the second step, a subset of the answer sets
is selected using a minimization procedure that examines the set of atoms that have
the value $F^*$ in each answer set. These two steps are formally defined below.
\begin{definition}\label{four-valued-preceq}
The ordering $\prec$ on truth values is defined as follows: ${F \prec F^*}$, ${F \prec T^*}$,
${F \prec T}$, and ${T^* \prec T}$. Given two truth values $v_1,v_2$, we write
${v_1 \preceq v_2}$ if either ${v_1 \prec v_2}$ or ${v_1 = v_2}$. Given
interpretations $I_1,I_2$ of a program $P$, we write $I_1 \preceq I_2$
if for all atoms $A$ in $P$, $I_1(A) \preceq I_2(A)$. We write $I_1 \prec I_2$
if $I_1 \preceq I_2$ but $I_1 \neq I_2$.
\end{definition}
It is easy to verify that $\preceq$ is a partial order.
\begin{definition}
An interpretation $I$ of LPOD $P$ is called \emph{solid} if for all atoms
$A$ in $P$, it is $I(A) \neq T^*$.
\end{definition}
\begin{definition}\label{logical-characterization-theorem}\label{answer-set}
An interpretation $M$ of an LPOD $P$ will be called an \emph{answer set} of $P$
if $M$ is a \mbox{$\preceq$-minimal} model of $P$ and $M$ is solid.
\end{definition}

\begin{definition}\label{sqsubseteq-ordering}
Let $P$ be an LPOD and let $M_1,M_2$ be answer sets of $P$. Let $M_1^*$ and
$M_2^*$ be the sets of atoms in $M_1$ and $M_2$ respectively that have the
value $F^*$. We say that $M_1$ is preferred to $M_2$, written
$M_1 \sqsubset M_2$, if $M_1^* \subset M_2^*$.
\end{definition}
\begin{definition}\label{most-preferred}
An answer set of an LPOD $P$ is called \emph{most-preferred} if it is minimal
among all the answer sets of $P$ with respect to the $\sqsubset$ relation.
\end{definition}
\begin{example}[taken from the paper by~\citeN{iclp2021}]\label{mercedes-vs-bmw}
Consider the following program whose declarative reading is \qemph{I prefer to buy
a Mercedes than a BMW. In case a Mercedes is available, I prefer a gas model to a
diesel one. A gas model of Mercedes is not available}.
\[
  \begin{array}{l}
  \mbox{\tt mercedes $\times$ bmw $\leftarrow$}\\
  \mbox{\tt gas\_mercedes $\times$ diesel\_mercedes $\leftarrow$ mercedes}\\
  \mbox{\tt false $\leftarrow$ gas\_mercedes, not false}
  \end{array}
\]
The last clause is a standard technique in ASP in order to state that an atom
({\tt gas\_mercedes} in our case) is not true. The
above program has two answer sets, namely:
\[
\begin{array}{l}
\{({\tt mercedes},T),({\tt bmw},F),({\tt gas\_mercedes},F^*),
        ({\tt diesel\_mercedes},T),({\tt false},F^*)\}\\
\{({\tt mercedes},F^*),({\tt bmw},T),({\tt gas\_mercedes},F^*),
        ({\tt diesel\_mercedes},F^*),({\tt false},F^*)\}
\end{array}
\]
According to the $\sqsubset$~ordering, the most-preferred answer set is the first
one because it minimizes the $F^*$ values. It is worth noting that under the
original semantics of LPODs~\cite{lpod-brewka,lpod-BNS04} two answer sets are
produced that are incomparable (and therefore they are both considered as
``most-preferred'').
\end{example}

%

\section{A Logical Characterization of Strong Equivalence for LPODs}\label{characterization}
In this section we establish a new, purely logical characterization of strong
equivalence for LPODs. Our investigation has as a starting point the work
of \citeN{FaberTW08}, in which an exhaustive study of different forms of strong
equivalence for LPODs was performed. Not all forms of strong equivalence studied
by \citeN{FaberTW08} are applicable in our case. An explanation of this state of
affairs and a detailed comparison of our technique with that
of \citeN{FaberTW08}, is given in Section~\ref{related}. In our work we examine
two notions of strong equivalence, namely \emph{strong equivalence under the
most-preferred answer sets}, and \emph{strong equivalence under all the answer sets}\footnote{These
two notions roughly correspond to the relations $\equiv^{i}_{s,\times}$ and
$\equiv_{s,\times}$ defined in the paper by \citeN{FaberTW08}.}.
We demonstrate that these notions can be captured by establishing logical
equivalence in the four-valued logic of Section~\ref{background} of the programs
involved.
\begin{definition}
Two LPODs $P_1$ and $P_2$ are termed \emph{strongly equivalent under the
most-preferred answer sets} if for every LPOD $P$, $P_1 \cup P$ and $P_2\cup P$ have
the same most-preferred answer sets.
\end{definition}
\begin{theorem}\label{inclusion-preference-theorem}
Two LPODs $P_1$, $P_2$ are strongly equivalent under the most-preferred answer sets
if and only if they are logically equivalent in four-valued logic.
\end{theorem}
\begin{proof}
\noindent ($\Leftarrow$)
Assume that $P_1$ and $P_2$ are logically equivalent in four-valued logic. Then,
every four-valued model that satisfies one of them, also satisfies the other. This
means that for all programs $P$, $P_1 \cup P$ has the same models as
$P_2 \cup P$. But then, $P_1 \cup P$ has the same most-preferred answer sets as
$P_2 \cup P$ (because the most-preferred answer sets of a program depend only on the
set of all the models of the program). Therefore, $P_1 \cup P$ and $P_2 \cup P$
are strongly equivalent under the most-preferred answer sets.

\noindent ($\Rightarrow$)
Assume that $P_1$ and $P_2$ are strongly equivalent. Suppose that $P_1$ has a
model $M$ which is not a model of $P_2$.
Without loss of generality, we may assume that $M(A)=F$, for every
atom $A$ in $\Sigma$ that does not occur in $P_1 \cup P_2$.

We will show that we can
construct an interpretation $M'$ and a program $P$ such that $M'$ is a
most-preferred answer set of one of $P_1 \cup P$ and $P_2 \cup P$ but not of the other,
contradicting our assumption of strong equivalence.

First, we construct two sets of atoms that will help us define $P$. In
particular, we construct two sets of atoms $\mathcal{T}$ and $\mathcal{F}$ each
one containing a new atom for every $A$ in $P$ such that $M(A) = F^*$. More
formally, let $\mathcal{T} = \{ t_A \mid M(A) = F^* \}$ and
$\mathcal{F} = \{ f_A \mid M(A) = F^* \}$, where all $t_A$ and $f_A$ do not
appear in $P_1$ and $P_2$.
We define $M'$ as:
\[
  M'(A) = \begin{cases}
              T      & M(A) = T^* \\
              T      & A \in \mathcal{T} \\
              F^*    & A \in \mathcal{F} \\
              M(A)   & \text{otherwise}
          \end{cases}
\]
%
%
%
We claim that $M'$ is a model of $P_1$. To verify this, take any rule in $P_1$ of the form
\[
  C_1 \times \cdots \times C_n \leftarrow A_1,\ldots,A_m,{\pnot B_1},\ldots,{\pnot B_k}
\]
If $M(A_1, \ldots, A_m, \pnot B_1, \ldots, \pnot B_k) \geq T^*$, then it is also
$M(C_1 \times \cdots \times C_n) \geq T^*$, since $M$ is a model of $P_1$. Then, there
exists $j \leq n$ such that $M(C_i) = F^*$ for all $i<j$, and $M(C_j) \geq T^*$. It follows
that $M'(C_i) = F^*$ for all $i<j$, and $M'(C_j) = T$, which implies $M'(C_1 \times \cdots \times C_n) = T$.
Therefore, $M'$ satisfies the rule in this case.

If $M(A_1, \ldots, A_m, \pnot B_1, \ldots, \pnot B_k) = F^*$, then there exists
$A_i$ such that $M(A_i) = F^*$. By the definition of $M'$, $M'(A_i) = F^*$, and thus
$M'(A_1, \ldots, A_m, \pnot B_1, \ldots, \pnot B_k) \leq F^*$.
Since $M$ is a model of $P_1$ it satisfies the given rule and thus
$M(C_1 \times \cdots \times C_n) \geq F^*$. If $M(C_1 \times \cdots \times C_n) = F^*$,
then for all $C_i$, $M(C_i)=F^*$, which implies that $M'(C_i)=F^*$ and therefore
$M'(C_1 \times \cdots \times C_n)=F^*$. If $M(C_1 \times \cdots \times C_n) > F^*$
then there exists $j\leq n$ such that for all $i<j$, $M(C_i)=F^*$ and $M(C_j)>F^*$,
which implies that for all $i<j$, $M'(C_i)=F^*$ and $M'(C_j)>F^*$, and therefore
$M'(C_1 \times \cdots \times C_n)>F^*$. In both cases, $M'$ satisfies the given rule.

If $M(A_1, \ldots, A_m, \pnot B_1, \ldots, \pnot B_k) = F$, then either there exists
$A_i$ such that $M(A_i) = F$ or there exists $B_j$ such that $M(B_j) \geq T^*$.
It follows, by the definition of $M'$, that $M'(A_i) = F$ or $M'(B_j) = T$ and
as a result $M'(A_1, \ldots, A_m, \pnot B_1, \ldots, \pnot B_k) = F$.
Therefore $M'$ satisfies the rule in this case. Thus, $M'$ is a model of $P_1$.

We proceed by distinguishing two cases that depend on
whether $M'$ is a model of $P_2$ or not.

\smallskip\noindent
\underline{\emph{Case 1}}: $M'$ is not a model of $P_2$. We take:
    \[P =  \{ A\leftarrow \mid M'(A) = T \} \cup
           \{ A \times t_A\leftarrow \mid t_A \in {\cal T} \} \cup
          \{ f_A \leftarrow \pnot f_A, A \mid f_A \in {\cal F} \} \]

We claim that every model $N$ of
  $\{ A \times t_A \leftarrow \mid t_A \in {\cal T} \} \cup
          \{ f_A \leftarrow \pnot f_A, A \mid f_A \in {\cal F} \} $
has the following property:
\begin{equation} \label{prop:1} \tag{P1}
  \textup{for every atom } A, \textup{ if } M'(A) = F^*, \textup{ then } N(A) \neq F
\end{equation}
In order to prove our claim we distinguish two cases for atoms such that
$M'(A) = F^*$: the atoms where $M(A) = F^*$ and the atoms in $\mathcal{F}$. For the
first case, assume that for some $A$ it is $M(A) = F^*$ and $N(A) = F$. But then
there exists a rule $A \times t_A\leftarrow$ in $P$ which is not satisfied by
$N$, which is a contradiction. So, for all such atoms $A$ it should be
$N(A) \geq F^*$. For the second case, assume that for some $f_A$, it is $N(f_A) = F$.
Then, the rule $f_A \leftarrow \pnot f_A, A$ is not satisfied by $N$ (since
$N(A) \geq F^*$), which is also a contradiction. Therefore, our claim holds.

%
Now, it is easy to see that $M'$ is a model of $P$ and therefore a model of
$P_1 \cup P$. Moreover, it is a most-preferred answer set of $P_1 \cup P$. Indeed, let $N$
be a model of $P_1 \cup P$ and $N \prec M'$. Since $M'$ does not assign any
$T^*$, there exists $A$ such that either $M'(A) = T$ and $N(A) \prec T$ or
$M'(A) = F^*$ and $N(A) = F$. In the first case, $P$ contains a fact
$A \leftarrow$, which is not satisfied by $N$. In the second case, $N$ does not
satisfy property~\ref{prop:1}. In both cases, $N$ is not a model of $P$, which is a
contradiction. It follows that $M'$ is $\preceq$-minimal model of $P_1 \cup P$.

Assume now that there exists some $N$ which is a most-preferred answer set of $P_1 \cup P$
and $N \sqsubset M'$. There must exist some atom $A$ such that $M'(A) = F^*$ and
$N(A) \neq F^*$. We will show in the following that it should be $N(A) = T$ for
those atoms.
First, notice that $N(A) \neq T^*$ because since $N$ is most-preferred it is
also solid. Also, it must be $N(A) \neq F$, since $N$ is a model of $P$ and thus
it satisfies property~\ref{prop:1}. Therefore, $N(A) = T$ for the atoms such that
$M'(A) = F^*$ and $N(A) \neq F^*$.

However, we now claim that $N$ is not $\preceq$-minimal. Indeed, we can construct
$N'$ from $N$ that is also a model of $P_1 \cup P$ and $N' \prec N$. Define $N'$ as:
\[
  N'(A) = \begin{cases}
            T^* & \text{if $A \in {\cal F}$ and $N(A) = T$} \\
            N(A) & \text{otherwise}
          \end{cases}
\]
First, we need to establish that $N' \prec N$, that is, there exists atom $A$
such that $N'(A) \prec N(A)$. By the assumption $N \sqsubset M'$ we know that there exists atom
$A$ such that $N(A) \neq M'(A)$ and we have established that $M'(A) = F^*$ and
$N(A) = T$. The first case is for the atom $A$ to be $A=f_B \in {\cal F}$ for
some $B$ and it is straightforward that $N(f_B) = T$ and $N'(f_B) = T^*$.
The second case is for $A$ to be an atom such that $M(A) = F^*$;
then there exists a rule $f_A \leftarrow \pnot f_A, A$ in $P$ which must be
satisfied by $N$. But since $N(A) = T$, and $N$ is solid, the only way to
satisfy this rule is when $N(f_A) = T$. By definition of $N'$, $N'(f_A) = T^*$.
Therefore, $N' \prec N$. It is also easy to see that $N'$ is a model of
$P_1 \cup P$ because it satisfies all rules $f_A \leftarrow \pnot f_A, A$ and
$f_A$ does not occur in any other rule of $P_1 \cup P$.

Therefore, $M'$ is a most-preferred answer set of $P_1 \cup P$.
This contradicts the assumption of strong equivalence because $M'$ is not
even a model for $P_2 \cup P$.

\smallskip\noindent
\underline{\emph{Case 2}}: $M'$ is a model of $P_2$.
Let $D$ be an atom in $\Sigma - ({\cal T} \cup {\cal F})$ that does not occur
in $P_1 \cup P_2$. Such an atom always exists, since $\Sigma$ is a countably
infinite set and ${\cal T}$, ${\cal F}$, $P_1$, and $P_2$ are finite; moreover,
$M(D) = F$, by our assumption about $M$.
We take:
\[
  \begin{split}
  P = & \{ A\leftarrow\ \mid M(A) = T \} \, \cup\\
      & \{ A \times t_A\leftarrow\ \mid  t_A \in {\cal T} \} \cup
        \{ f_A \leftarrow \pnot f_A, A \mid  f_A \in {\cal F} \}\, \cup \\
      & \{ B \leftarrow A \mid A \neq B \text{ and } M(A) = M(B) = T^* \}\, \cup \\
      & \{ D \leftarrow \pnot A \mid M(A) = T^* \}
  \end{split}
\]
It is easy to see that $M'$ satisfies every formula in $P$, and therefore it is
a model of both $P_1 \cup P$ and $P_2 \cup P$. We will show that $M'$ is a
most-preferred answer set of $P_2 \cup P$ but not a most-preferred answer set of $P_1 \cup P$.

We proceed by showing that $M'$ is a $\preceq$-minimal model of $P_2 \cup P$.
Assume there exists a model $N$ of $P_2 \cup P$ such that $N \prec M'$.

%
We first show that there exists an atom $A$ such that $M(A) = T^*$ and $N(A) = T$.
Consider an arbitrary atom $C$.
If $M(C)=T$, then it is also $N(C)=T$, because $P$ contains $C \leftarrow$ and $N$ is a
model of $P$.
If $M(C)=F^*$, then by the construction of $M'$ it is $M'(C)=F^*$.
Since $N$ is a model of $P$, by property~\ref{prop:1} we obtain $N(C) \neq F$. This implies
$N(C) = F^*$, because $N \prec M'$.
If $M(C) = F$, then by the construction of $M'$ it is $M'(C) = F$, and since $N \prec M'$ we get $N(C) = F$.
Therefore, if $M(C) \neq T^*$, then $M(C)=N(C)$.
There should be, however, an atom $A$ that occurs in $P_2$ such that $N(A) \neq M(A)$
because $N$ is a model of $P_2$ and $M$ is not.
Obviously, for that atom it must be $M(A) = T^*$ and $N(A) \neq T^*$.
Now, notice that there exists a rule $D \leftarrow \pnot A$ in $P$
where $M(D) = F$ and must be satisfied by $N$ since it is also a
model of $P$. Since $M(D) = F$ implies
$N(D) = F$, the only remaining possibility is $N(A) = T$.

We next show that there exists an atom $B$ such that $M(B) = N(B) = T^*$.
Since $N \prec M'$, there exists $B$ such that $N(B) \prec M'(B)$.
The last relation immediately implies $M'(B) \neq F$. Notice also that, by the construction of $M'$,
it is $M'(B)\neq T^*$. Moreover, it cannot be $M'(B) = F^*$, since in that case
from $N(B) \prec M'(B)$ we would obtain $N(B) = F$, which contradicts property~\ref{prop:1}.
Therefore, the only remaining value is $M'(B) = T$.
For that atom, it cannot be $M(B) = T$ because then it would also be $N(B) = T$
(since $B \leftarrow$ is a rule in $P$ and $N$ is a model of $P$), which would
contradict $N(B) \prec M'(B)$. It follows by the construction of $M'$ that $M(B) = T^*$.
We claim that $N(B) = T^*$, that is, it cannot be $N(B) = F$.
Since $M(B) = T^*$ there exists a rule $D \leftarrow \pnot B$ where $M(D) = F$.
Since, $M(D) = F$, we get $N(D) = F$. If we assume
that also $N(B) = F$ then $N$ does not satisfy this rule which is a
contradiction. Therefore, $N(B) = T^*$.

Since $M(A) = M(B) = T^*$ there exists rule $B \leftarrow A$ in $P$ that is not
satisfied by $N$ because we have showed that $N(B) = T^*$ and $N(A) = T$.
Therefore, $N$ is not a model of $P_2 \cup P$, which is a contradiction.

We conclude that $M'$ is $\preceq$-minimal model of $P_2 \cup P$.
Following an identical reasoning as in the final paragraph of the proof of Case 1,
we can show that $M'$ is a most-preferred answer set of $P_2 \cup P$. In order to conclude
the proof, it suffices to show that $M'$ is not a most-preferred answer set of $P_1 \cup P$.
We define $M''$ as:
\[
  M''(A) = \begin{cases}
              T      & A \in \mathcal{T} \\
              F^*    & A \in \mathcal{F} \\
              M(A)   & \text{otherwise}
          \end{cases}
\]
$M''$ is not a model of $P_2$ because $M$ is not a model of $P_2$.
By definition, $M'' \preceq M'$. But $M'' \neq M'$ because $M'$ is a model of
$P_2$ and $M''$ is not. Therefore, $M'' \prec M'$. Observe that $M''$ agrees
with $M$ for all $A$ that appear in $P_1$ and since $M$ is a model of $P_1$,
$M''$ is also a model of $P_1$. $M''$ also satisfies the rules of $P$ and
therefore it is a model of $P_1 \cup P$. Therefore, $M'$ is not a most-preferred
answer set of $P_1 \cup P$.
\end{proof}

We now consider the second notion of strong equivalence that is applicable in
our setting.
\begin{definition}
Two LPODs $P_1$, $P_2$ are termed \emph{strongly equivalent
under all the answer sets}, if for every LPOD $P$, $P_1 \cup P$ and
$P_2 \cup P$ have the same answer sets.
\end{definition}
\begin{theorem}\label{all-answer-sets-theorem}
Two LPODs $P_1$, $P_2$ are strongly equivalent under all the answer sets
if and only if they are logically equivalent in four-valued logic.
\end{theorem}
The proof of the above theorem, which can be found in~\ref{appendix1}, follows
the same steps as that of the proof of Theorem~\ref{inclusion-preference-theorem},
omitting the parts of the proof related to $\sqsubset$-minimization.
\begin{corollary}
Two LPODs $P_1$, $P_2$ are strongly equivalent under the most-preferred answer sets
if and only if they are strongly equivalent under all the answer sets.
\end{corollary}
We feel that the above corollary highlights an interesting fact: it states that assessing
the observable behaviour of two programs with respect to the most-preferred answer sets,
suffices to determine strong equivalence of the programs.

Due to the above corollary, in the following we will often talk
about ``strong equivalence of LPODs'' without specifying the exact type of equivalence (since they
coincide).
\begin{example}
One can easily verify (using a four-valued truth table or a case analysis) that
the programs:
\[
  \begin{array}{l}
  \mbox{\tt a $\times$ b $\leftarrow$}\\
  \mbox{\tt a $\leftarrow$}
  \end{array}
\]
and the program that consists of just the following fact:
\[
  \begin{array}{l}
  \mbox{\tt a $\leftarrow$}
  \end{array}
\]
are strongly equivalent. Similarly, one can verify that the programs given in
Example 3 of the paper by \citeN{FaberTW08}, namely:
\[
  \begin{array}{l}
  \mbox{\tt c $\times$ a $\times$ b $\leftarrow$}\\
  \mbox{\tt a $\leftarrow$ c}\\
  \mbox{\tt b $\leftarrow$ c}\\
  \mbox{\tt c $\leftarrow$ a,b}
  \end{array}
\]
and:
\[
  \begin{array}{l}
  \mbox{\tt c $\times$ a $\times$ b $\leftarrow$}\\
  \mbox{\tt c $\times$ c $\times$ b $\times$ a $\leftarrow$}\\
  \mbox{\tt a $\leftarrow$ c}\\
  \mbox{\tt b $\leftarrow$ c}\\
  \mbox{\tt c $\leftarrow$ a,b}
  \end{array}
\]
are also strongly equivalent. Notice that the above two programs are also strongly equivalent
under the relations $\equiv^{i}_{s,\times}$ and $\equiv_{s,\times}$ defined in the paper by
\citeN{FaberTW08} (see the discussion in Example 3, page 441, of the aforementioned paper).
\end{example}
We now demonstrate that our characterization of strong equivalence, when restricted to normal
logic programs, retains the spirit of the initial characterization of strong equivalence for such programs
\footnote{Actually, the syntax of the programs treated in the paper by \citeN{LifschitzPV01},
is broader than that of normal logic programs.}~\cite{LifschitzPV01}. More specifically, we show that in order to characterize
strong equivalence for normal logic programs, it suffices to look at their models that contain
only the truth values $F$, $T^*$, and $T$.

We define strong equivalence for normal programs in the standard way~\cite{LifschitzPV01}.
The ``standard answer set semantics'' is the usual stable model semantics~\cite{GL88} of
normal logic programs.
\begin{definition}
Two normal logic programs $P_1$ and $P_2$ are termed \emph{strongly equivalent under the
standard answer set semantics}, if for every normal logic program $P$, $P_1 \cup P$
and $P_2 \cup P$ have the same standard answer sets.
\end{definition}
The following definition and theorem characterize strong equivalence of normal programs
in our setting.
\begin{definition}
An interpretation $I$ of an LPOD $P$ is called \emph{three-valued} if for all
atoms $A$ in $P$, it is $I(A) \neq F^*$. A \emph{three-valued model} of $P$ is a three-valued
interpretation of $P$ that is also a model of $P$.
\end{definition}
\begin{theorem}\label{backwards-theorem}
Let $P_1$, $P_2$ be normal logic programs. Then, $P_1$ and $P_2$ are strongly
equivalent under the standard answer set semantics if and only if they have
the same three-valued models.
\end{theorem}
The proof of the above theorem is given in~\ref{appendix1}.

\section{The Complexity of Strong Equivalence for LPODs}\label{complexity}
In this section we examine the complexity of strong equivalence under our new
characterization. Since the two versions of strong equivalence that we have examined
have an identical characterization (see Theorems~\ref{inclusion-preference-theorem} and~\ref{all-answer-sets-theorem}), the same complexity applies in both cases.

Our proof establishes \coNP-hardness by a direct (and quite simple) reduction
from 3SAT, which uses the special structure of LPODs in a crucial way. The corresponding
proof by \citeN{FaberTW08} utilizes the more involved \coNP-hardness
result of~\citeN{Lin02} for strong equivalence of normal logic programs~\footnote{As remarked
by one of the reviewers, the \coNP-completeness of strong equivalence for standard ASP programs was
first shown in the paper by \citeN{PearceTW01}.}. In this respect,
we feel that the proof that follows, apart from the fact that it applies to our new
characterization of strong equivalence, also has an interest in its own right due
to its different approach.
\begin{theorem}
Strong equivalence of LPODs
is a \coNP-complete problem.
\label{thm-coNP}
\end{theorem}

\begin{proof}
Let $P_1,P_2$ be two LPODs that are not strongly equivalent.
Then, without loss of
generality, there exists a four-valued interpretation $I$ that is a model of
$P_1$, but not a model of $P_2$.
Assume that the ground atoms that occur in $P_1 \cup P_2$ are $A_1, A_2, \dots, A_m$, and
consider the certificate $\mathcal{C}=[A_1,I(A_1),A_2,I(A_2), \dots, A_m,I(A_m)]$. $\mathcal{C}$
has size polynomial to the size of $(P_1, P_2)$; moreover, given $P_1$, $P_2$ and $\mathcal{C}$
it can be verified in polynomial time that $P_1$ and $P_2$ are not strongly equivalent.
Thus, deciding whether two programs are strongly equivalent
is in \coNP.

We next prove that strong equivalence of LPODs
is also a \coNP-hard problem, using a polynomial time reduction of 3SAT to the complement
of this problem.

Let $\phi= \bigwedge_{i=1}^n c_i$ be a propositional formula in conjunctive normal form,
where $c_i = L_{i,1} \vee L_{i,2} \vee L_{i,3}$ and $L_{i,j}$ is a literal
(that is, either a variable or the negation of a variable). For convenience, we may assume that
the variables that occur in $\phi$ are elements of $\Sigma$.


We will construct two programs $P_1$, $P_2$, such that $\phi$ is satisfiable if and only if
$P_1$ and $P_2$ are not strongly equivalent.

For every literal $L$ we define $\widetilde{L}$ as follows:
\[
             \widetilde{L} = \left\{
                             \begin{array}{ll}
                             L & \mbox{if $L = C$, for some $C \in \Sigma$}\\
                             not\ C & \mbox{if $L = \neg C$, for some $C \in \Sigma$}\\
                             \end{array}
                      \right.
       \]
Let $A$, $B$ be two propositional variables in $\Sigma$ that do not occur in $\phi$ and let $Q$ be
\[Q = \{ A \leftarrow \widetilde{L_{i,1}},\widetilde{L_{i,2}},\widetilde{L_{i,3}} \mid 1 \leq i \leq n\}\]
The LPODs $P_1$ and $P_2$ are defined as follows:
\[ P_1 = Q \cup \{ A \times B \leftarrow \} \]
\[ P_2 = Q \cup \{ A \times B \leftarrow \} \cup \{ A \leftarrow \}\]

Assume that $\phi$ is satisfiable and let $J$ be a two-valued interpretation such that $J(\phi) = T$.
We define the four-valued interpretation $I$ as follows:
\begin{eqnarray*}
I(A) & = & F^* \\
I(B) & = & T \\
I(C) & = & F, \textup{ if } C \textup{ occurs in } \phi \textup{ and } J(C)=T \\
I(C) & = & T, \textup{ if } C \textup{ occurs in } \phi \textup{ and } J(C)=F
\end{eqnarray*}

Consider an arbitrary rule $A \leftarrow \widetilde{L_{i,1}},\widetilde{L_{i,2}},\widetilde{L_{i,3}}$
in $Q$.
Then, $L_{i,1} \vee L_{i,2} \vee L_{i,3}$ is a clause in $\phi$;
since $J$ satisfies $\phi$, it holds $J(L_{i,j})=T$, for some $j\ \in \{1,2,3\}$.
Therefore, $I(\widetilde{L_{i,j}})=F$, which implies that $I$ satisfies the rule
$A \leftarrow \widetilde{L_{i,1}},\widetilde{L_{i,2}},\widetilde{L_{i,3}}$.
Moreover, $I(A \times B) = T$. We conclude that $I$ is a model of $P_1$; however,
$I$ is not a model of $P_2$, since $I(A)=F^*$.
Therefore, $P_1$ and $P_2$ are not logically equivalent in the four-valued logic,
which implies that they are not strongly equivalent.

Conversely, assume that $P_1$ and $P_2$ are not strongly equivalent.
Then, $P_1$ and $P_2$ are not logically equivalent in our four-valued logic.
Since $P_1 \subset P_2$, there exists a four-valued interpretation $I$ that is a model of $P_1$,
but not a model of $P_2$. We define the following two-valued interpretation for the variables in $\phi$:
\[
             J(C) = \left\{
                             \begin{array}{ll}
                             T & \mbox{if $I(C) \leq F^*$}\\
                             F & \mbox{if $I(C) \geq T^*$}\\
                             \end{array}
                      \right.
\]
We will show that $J$ satisfies $\phi$. We first prove some properties of $I$.

Since $I$ is a model of $P_1$, it must be either $I(A)=T$, or $I(A)=F^*$ and
$I(B)=T$, so that the rule $A \times B\leftarrow$ is satisfied. However, in the former
case, $I$ should also be a model of $P_2$ (since $P_2-P_1 = \{ A \leftarrow \}$), which is
a contradiction. Therefore, only the latter case is possible, that is,
$I(A)=F^*$.

Consider an arbitrary clause $c_i = L_{i,1} \vee L_{i,2} \vee L_{i,3}$ in $\phi$.
Since $I$ is a model of $P_1$, $I$ satisfies the rule
$A \leftarrow \widetilde{L_{i,1}},\widetilde{L_{i,2}},\widetilde{L_{i,3}}$ in $Q \subset P_1$.
Therefore, $\min\{I(\widetilde{L_{i,1}}),I(\widetilde{L_{i,2}}),I(\widetilde{L_{i,3}})\} \leq I(A) = F^*$,
which implies that there exists a $j \in \{1,2,3\}$ such that $I(\widetilde{L_{i,j}}) \leq F^*$.
But then, $J(L_{i,j})=T$. We conclude that $\phi$ is satisfiable.
\end{proof}

\section{Related and Future Work}\label{related}
The work on strong equivalence, started with the pioneering results of~\citeN{LifschitzPV01},
but has since been extended to various formal systems. In particular, strong equivalence
has been abstractly studied as a property across a variety of preferential formalisms~\cite{FaberTW13}.
To our knowledge however, the only existing work on the strong equivalence of LPODS
is the paper by \citeN{FaberTW08}. In that work the authors present an exhaustive study of
several notions of strong equivalence for LPODs. More specifically, given LPODs
$P$, $Q$, they consider the following notions of strong equivalence:
\begin{enumerate}
\item $P \equiv_s Q$ holds iff the standard answer sets of $P$ and $Q$ coincide
      under any extension by ordinary (namely, \emph{normal}) programs.

\item $P \equiv_{s,\times} Q$ holds iff the standard answer sets of $P$ and $Q$
      coincide under any extension by LPODs.

\item $P \equiv_s^{\sigma} Q$ holds iff the $\sigma$-preferred answer sets of
      $P$ and $Q$ coincide under any extension by ordinary programs, where
      $\sigma \in \{i,p,c\}$ and the indices $i$, $c$, and $p$ correspond to the
      \emph{inclusion}, \emph{Pareto}, and \emph{cardinality} orderings respectively
      (see the paper by~\citeN{lpod-BNS04} for formal definitions of these orderings).

\item $P \equiv_{s,\times}^{\sigma} Q$ holds iff the $\sigma$-preferred answer
      sets of $P$ and $Q$ coincide under any extension by LPODs, where
      $\sigma \in \{i,p,c\}$.
\end{enumerate}
Considering the above notions, the study of~\citeN{FaberTW08} is certainly
broader than the present work. We have not considered cases (1) and
(3) above because in the standard
definition of strong equivalence~\cite{LifschitzPV01} both the programs under
comparison and the context-programs, all belong to the same source language (in
our case, LPODs). Of course, there may exist application domains where relations
like $\equiv_s$ and $\equiv_s^{\sigma}$ might be of interest. In such a case, it
might prove interesting to extend the present work in this direction. Case (2)
above is covered by our Theorem~\ref{all-answer-sets-theorem}. Finally, from
case (4) above, we cover only the subcase where $\sigma$ is the \emph{inclusion}
preference. The subcases of Pareto and cardinality preferences are not covered
because the semantics of~\citeN{iclp2021} on which the present work is based, is
defined using the relation $\sqsubset$, which is the model-theoretic version of
the inclusion preference of~\cite{lpod-brewka,lpod-BNS04}. It is important,
however, to stress that the inclusion preference is probably the most fundamental
among the three orderings and the initial paper introducing
LPODs~\cite{lpod-brewka}, used only this one. The Pareto and
cardinality preferences were proposed subsequently in order to remedy the
shortcomings of the initial semantics of LPODs~\cite[see the discussion in page 342]{lpod-BNS04}.
Notice also that the cardinality preference can not be
generalized in a direct way to first-order programs whose ground instantiation
consists of an infinite number of rules.

Recapitulating, the two notions of strong equivalence that we cover in the present
paper (Theorems~\ref{inclusion-preference-theorem} and~\ref{all-answer-sets-theorem}),
correspond to the relations $\equiv^{i}_{s,\times}$ and $\equiv_{s,\times}$
defined in the paper by \citeN{FaberTW08}. In our case, both notions of strong equivalence coincide,
because they have a  unique characterization as logical equivalence in our four-valued logic.
On the other hand, the relations $\equiv^{i}_{s,\times}$ and $\equiv_{s,\times}$ do not coincide
\cite[see Theorem 21]{FaberTW08}. This means that our approach and that of \citeN{FaberTW08}
are different: there exist programs that are strongly equivalent with respect to one of the
approaches and not strongly equivalent with respect to the other approach. This was expected since
the two approaches are based on markedly different semantics. Although it does not seem straightforward
to establish a formal relation between our framework and that of \citeN{FaberTW08}, we can find examples
where the two approaches give different results.
\begin{example}
Consider the following two programs given in Example 2 of the paper by \citeN{FaberTW08}:
\[
  \begin{array}{l}
  \mbox{\tt c $\times$ a $\times$ b $\leftarrow$}\\
  \mbox{\tt c $\leftarrow$ a,b}\\
  \mbox{\tt d $\leftarrow$ c,not d}
  \end{array}
\]
and:
\[
  \begin{array}{l}
  \mbox{\tt c $\times$ b $\times$ a $\leftarrow$}\\
  \mbox{\tt c $\leftarrow$ a,b}\\
  \mbox{\tt d $\leftarrow$ c,not d}
  \end{array}
\]
It is intuitively clear that in the first program {\tt a} is preferred over {\tt b}, while in the
second program {\tt b} is preferred over {\tt a}. Despite this difference, the two programs are strongly
equivalent under the $\equiv_{s,\times}$ semantics of~\cite{FaberTW08}. Under our characterization the
two programs are not strongly equivalent. To see this, consider the interpretation
$I=\{({\tt a},T),({\tt b},F),({\tt c},F^*),({\tt d},F^*)\}$, which is a model of the first program
but not a model of the second. Therefore, the two programs are not logically equivalent in our four-valued
logic, and consequently they are not strongly equivalent in our setting.
\end{example}

Although our study does not cover all the notions of strong equivalence examined
in~\cite{FaberTW08}, we believe that it has important advantages. Our work characterizes
strong equivalence as logical equivalence in the four-valued logic of~\cite{iclp2021}.
This result extends in a smooth way the well-known characterization of strong equivalence
for normal logic programs~\cite{LifschitzPV01}. Notice that the corresponding
characterization of the inclusion preferred strong equivalence
in~\cite{FaberTW08}, is much more involved and uses certain binary
functions over the sets of models of the programs that rely on the syntax
of the given programs (see~\cite{FaberTW08}, Definition 8 and Theorem 19).
We believe that this is not an inherent shortcoming of the work of~\cite{FaberTW08},
but instead a possibly unavoidable consequence of the fact that the original
semantics of LPODs~\cite{lpod-brewka,lpod-BNS04} is not purely model theoretic.
The simplicity of our characterization makes us believe that it can be extended
to broader classes of programs, such as for example to LPODs with strong negation
and to disjunctive LPODs~\cite{iclp2021}.

One important aspect that we have not examined in this paper, is the possible
practical use of the proposed strong equivalence characterization. To our knowledge,
all major ASP systems are two-valued, and it is therefore a legitimate question of how
our four-valued framework can be embedded in such systems. We believe that a promising
direction for future work would be to define a notion of \emph{collapsed strong
equivalence} for LPODs:
\begin{definition}
Two LPODs $P_1$ and $P_2$ are termed \emph{collapsed strongly equivalent under
the most-preferred answer sets} if for every LPOD $P$, the most-preferred answer sets of
$P_1 \cup P$ and $P_2\cup P$ become identical when $F^*$ is collapsed to $F$.
\end{definition}
Notice that in the above definition we do not need to collapse $T^*$ to $T$
because, by Definition~\ref{answer-set}, answer sets do not contain the $T^*$
value. The logical characterization of collapsed strongly equivalent LPODs is
probably an interesting question that deserves further investigation.

Finally, a very interesting question raised by one of the reviewers, is whether
the techniques developed in the paper by \citeN{iclp2021}, can also be used to derive a novel
and simpler semantics for Qualitative Choice Logic (QCL)~\cite{qcl}. Notice that QCL has also recently been
investigated with respect to strong equivalence~\cite{Bernreiter0W21}, so the work
developed in the present paper may be also relevant in this more general context.

\bibliography{lpods}

\ifincludeappendix
\appendix
\clearpage
\input{strong_proofs}

\ifreview
\clearpage
\input{response}
\fi

\fi

\end{document}

%% file: strong_proofs.tex
\section{Proofs of Theorem~\ref{all-answer-sets-theorem} and Theorem~\ref{backwards-theorem}}\label{appendix1}
This appendix contains the proofs of Theorems~\ref{all-answer-sets-theorem}
and~\ref{backwards-theorem} from Section~\ref{characterization}. 
\begin{retheorem}{all-answer-sets-theorem}
Two LPODs $P_1$, $P_2$ are strongly equivalent under all the answer sets
if and only if they are logically equivalent in four-valued logic.
\end{retheorem}
\begin{proof}
\noindent ($\Leftarrow$)
Assume that $P_1$ and $P_2$ are logically equivalent in four-valued logic. Then,
every four-valued model that satisfies one of them, also satisfies the other. This
means that for all programs $P$, $P_1 \cup P$ has the same models as
$P_2 \cup P$. But then, $P_1 \cup P$ has the same answer sets as $P_2 \cup P$
(because the answers sets of a program are the $\preceq$-minimal models
among all the models of the program). Therefore, $P_1 \cup P$ and $P_2 \cup P$
are strongly equivalent under all the answer sets.

\noindent ($\Rightarrow$)
Assume that $P_1$ and $P_2$ are strongly equivalent under all the answer sets.
Assume, for the sake of contradiction, that $P_1$ has a model $M$ which is not a model
of $P_2$. We will show that we can construct an interpretation $M'$ and a program $P$ such
that $M'$ is a $\preceq$-minimal model of one of $P_1 \cup P$ and $P_2 \cup P$
but not of the other, contradicting our assumption of strong equivalence under all the 
answer sets. The construction of $M'$ and the proof that $M'$ is a model of $P_1$, are
identical to the corresponding ones in the proof of
Theorem~\ref{inclusion-preference-theorem}. We distinguish two cases.

\smallskip\noindent
\underline{\emph{Case 1}}: $M'$ is not a model of $P_2$. We define exactly the
same program $P$ as in Case 1 of Theorem~\ref{inclusion-preference-theorem} and
we demonstrate, following the same steps, that $M'$ is a $\preceq$-minimal model
of $P_1 \cup P$. This contradicts our assumption of strong equivalence because
$M'$ is not even a model of $P_2 \cup P$ (since we have assumed that it is not a
model of $P_2$).

\smallskip\noindent
\underline{\emph{Case 2}}: $M'$ is a model of $P_2$. We define exactly the same
program $P$ as in Case 2 of Theorem~\ref{inclusion-preference-theorem} and we
demonstrate, following the same steps, that $M'$ is a $\preceq$-minimal model of
$P_2 \cup P$. We then show, following the same steps as in
the proof of Theorem~\ref{inclusion-preference-theorem}, that $M'$ is not a
$\preceq$-minimal model of $P_1 \cup P$. This contradicts our assumption of
strong equivalence under all answer sets.

In conclusion, $P_1$ and $P_2$ are logically equivalent.
\end{proof}
For the proof of Theorem~\ref{backwards-theorem} we will make use of the following 
lemma from the paper by \citeN{iclp2021}:
\begin{lemma}\label{answer-sets-coincide}
Let $P$ be a normal logic program. Then, the answer sets of $P$ (see Definition~\ref{answer-set})
coincide with the standard answer sets of $P$.
\end{lemma}
\begin{retheorem}{backwards-theorem}
Let $P_1$, $P_2$ be normal logic programs. Then, $P_1$ and $P_2$ are strongly
equivalent under the standard answer set semantics if and only if they have
the same three-valued models.
\end{retheorem}
\begin{proof}
\noindent ($\Leftarrow$)
Assume that $P_1$ and $P_2$ have the same three-valued models.
This means that for all programs $P$, $P_1 \cup P$ has the same three-valued models
as $P_2 \cup P$. Since $P_1 \cup P$ and $P_2 \cup P$ are normal programs,
by Lemma~\ref{answer-sets-coincide} the answer sets coincide with the standard
answer sets which are two-valued by definition and therefore the answer sets are
the $\preceq$-minimal models among the three-valued models of the program.
But then, $P_1 \cup P$ has the same answer sets (and the same standard answer sets)
as $P_2 \cup P$. Therefore, $P_1$ and $P_2$ are strongly equivalent under
the standard answer set semantics.

\noindent ($\Rightarrow$)
Assume that $P_1$ and $P_2$ are strongly equivalent under the standard answer
set semantics. Suppose that $P_1$ has a three-valued model $M$ which is not a
model of $P_2$. Without loss of generality, we may assume that $M(A) = F$, for
every atom $A \in \Sigma$ that does not occur in $P_1 \cup P_2$. We will show
that we can construct an three-valued interpretation $M'$ and a normal logic
program $P$ such that $M'$ is a standard answer set of one of $P_1 \cup P$ and
$P_2 \cup P$ but not of the other contradicting our assumption of strong
equivalence.

Let $M'$ be the two-valued interpretation defined as:
\[
  M'(A) = \begin{cases}
              T & M(A) \geq T^* \\
              F & \text{otherwise}
            \end{cases}
\]

We claim that $M'$ is a model of $P_1$. Since $P_1$ is a normal logic program
all rules are of the form $C \leftarrow A_1, \ldots, A_m, {\pnot B_1}, \ldots,
{\pnot B_k}$. If $M'(A_1, \ldots, A_m, {\pnot B_1}, \ldots, {\pnot B_k}) = F$
then the rule is trivially satisfied. If $M'(A_1, \ldots, A_m, {\pnot B_1},
\ldots, {\pnot B_k}) = T$ then it follows that $M(A_i) \geq T^*$ and $M(B_j)= F$
for every $A_i$ and $B_j$ in the body of the rule and $M(A_1, \ldots,
A_m, {\pnot B_1}, \ldots, {\pnot B_k}) \geq T^*$. Since $M$ is a model of $P_1$
it satisfies the rule and thus $M(C) \geq T^*$. By the construction of $M'$ it
follows that $M'(C) = T$ and consequently the rule is satisfied. Lastly, notice
that no other values are possible for the body of the rule and therefore we
conclude that $M'$ is a model of $P_1$.

We proceed by distinguishing two cases that depend on whether $M'$ is a model
of $P_2$ or not.

\smallskip\noindent
\underline{\emph{Case 1}}: $M'$ is not a model of $P_2$.
We take $P$ to be $\{ A \leftarrow | M'(A) = T \}$. It is easy to see that $M'$
is a model of $P$ and thus model of $P_1 \cup P$. We show that $M'$ is also a
$\preceq$-minimal model of $P_1 \cup P$ and since $P_1 \cup P$ is a normal logic
program $M'$ is also a standard answer set of $P_1 \cup P$. Let $N$ be a model
of $P_1 \cup P$  and $N \prec M'$. It must exist atom $A$ such that $N(A) \prec
M'(A)$. Since $M'$ assigns only values $T$ and $F$, it must be $N(A) = F$ and $M(A) = T$. 
But then, $N$ is not a model of $P$ because there is a rule $A \leftarrow$ in $P$
which leads to contradiction. Therefore, $M'$ is $\preceq$-minimal and a
standard answer set of $P_1 \cup P$. By our initial assumption, $M'$ is not a
model of $P_2$ and thus not a model of $P_2 \cup P$ which leads to the
contradiction that $P_1$ and $P_2$ are strongly equivalent.

\smallskip\noindent
\underline{\emph{Case 2}}: $M'$ is a model of $P_2$. Let $D$ be an atom in $\Sigma$
that does not occur in $P_1 \cup P_2$. Such atom always exists, since $\Sigma$
is countably infinite set and $P_1, P_2$ are finite; moreover, $M(D) = F$ by our
assumption about $M$.
We take $P$ to be
\[
  \begin{split}
  P = & \{ A \leftarrow\ \mid M(A) = T \}\, \cup \\
      & \{ B \leftarrow A \mid \text{$A \neq B$ and $M(A) = T^*$ and $M(B) = T^*$ }\}\, \cup \\
      & \{ D \leftarrow \pnot A \mid M(A) = T^* \}
  \end{split}
\]
It is easy to see that $M'$ satisfies every rule in $P$ and therefore is a model
of both $P_1 \cup P$ and $P_2 \cup P$. We show that $M'$ is a standard answer
set of $P_2 \cup P$ but not of $P_1 \cup P$.

We proceed by showing that $M'$ is a $\preceq$-minimal model of $P_2 \cup P$ and
therefore an answer set of $P_2 \cup P$ which by
Lemma~\ref{answer-sets-coincide} is also a standard answer set of $P_2 \cup P$.
Assume there exists a model $N$ of $P_2 \cup P$ such that $N \prec M'$.

We first show that there exists an atom $A$ such that $M(A) = T^*$ and $N(A) = T$.
Consider an arbitrary atom $C$. If $M(C) = T$ then it is also $N(C) = T$, because
$P$ contains $C \leftarrow$ and $N$ is a model of $P$. If $M(C) = F$ then, by the
construction of $M'$ it is $M'(C) = F$ and since $N \prec M'$ we get $N(C) = F$.
Therefore if $M(C) \neq T^*$ then $M(C) = N(C)$. There should be, however, an atom $A$
that occurs in $P_2$  such that $N(A) \neq M(A)$ because $N$ is a model of $P_2$
and $M$ is not. Obviously, for that atom it must be $M(A) = T^*$ and $N(A) \neq T^*$.
Notice that there exists a rule $D \leftarrow \pnot A$ in $P$ where $M(D) = F$ and
must be satisfied by $N$ since it is also a model of $P$. Since $M(D) = F$ implies
$N(D) = F$, the only possibility is $N(A) = T$.

We next show that there exists an atom $B$ such that $M(B) = N(B) = T^*$.
Since $N \prec M'$, there exists $B$ such that $N(B) \prec M'(B)$.  The last relation
immediately implies that $M'(B) \neq F$ and by the construction of $M'$, it is
$M'(B) \neq T^*$. Therefore, the only remaining value is $M'(B) = T$. For that atom,
it cannot be $M(B) = T$ because then it is also $N(B) = T$. It follows, by the
construction of $M'$ that $M(B) = T^*$. We claim that $N(B) = T^*$, that is, it cannot
be $N(B) = F$. Since $M(B) = T^*$ there exists a rule $D \leftarrow \pnot B$ where
$M(D) = F$. Since $M(D) = F$, it is also $N(D) = F$. If we assume that $N(B) = F$
then $N$ does not satisfy this rule which is a contradiction. Therefore, $N(B) = T^*$.

Since $M(A) = M(B) = T^*$ there exists a rule $B \leftarrow A$ in $P$ that is not
satisfied by $N$ because we showed that $N(B) = T^*$ and $N(A) = T$. Therefore, $N$
is not a model of $P_2 \cup P$ and $M'$ is $\preceq$-minimal model of $P_2 \cup P$.

In order to conclude the proof, it suffices to show that $M'$ is not a standard
answer set of $P_1 \cup P$. By the definition of $M'$, it is $M \preceq M'$. But
since $M'$ is a model of $P_2$ and $M$ is not, it must be $M' \neq M$ and thus
$M \prec M'$. $M$ also satisfies the rules of $P$ and therefore it is a model of
$P_1 \cup P$. We conclude that $M'$ is not $\preceq$-minimal model of $P_1 \cup
P$ and thus not a standard answer set of $P_1 \cup P$.
\end{proof}

